\documentclass[10pt]{amsart}
\usepackage[cp866]{inputenc}
\usepackage[english,russian]{babel}
\usepackage{amsmath}
\usepackage{amssymb}
\usepackage{amsfonts}

\newtheorem{lemma}{Lemma}
\newtheorem{theorem}{Theorem}

\newtheorem{corollary}{Corollary}
\newtheorem{proposition}{Proposition}

\newcommand{\Aut}{\mathrm{Aut}}
\newcommand{\PAut}{\mathrm{PAut}}

\newcommand{\wt}{\mathrm{wt}}

\newcommand{\GL}{\mathrm{GL}}
\newcommand{\GA}{\mathrm{GA}}

\newcommand{\supp}{\mathrm{supp}}

\begin{document}
\renewcommand{\refname}{References}
\renewcommand{\proofname}{Proof.}
\renewcommand{\tablename}{Table}
\renewcommand{\dim}{\mathrm{dim}}

\title[Coordinate transitivity of extended perfect codes and their SQS
]{Coordinate transitivity of a class of\\ extended perfect codes
and their SQS
}
\author{{I. Yu. Mogilnykh, F. I. Solov'eva}}%

\address{Ivan Yurevich Mogilnykh
\newline\hphantom{iii} Tomsk State University, Regional Scientific and Educational Mathematical Center,
\newline\hphantom{iii} pr. Lenina, 36,
\newline\hphantom{iii} 634050, Tomsk, Russia,
\newline\hphantom{iii} Novosibirsk State University,
\newline\hphantom{iii} Pirogova street, 1,
\newline\hphantom{iii} 630090, Novosibirsk, Russia
\newline\hphantom{iii} Sobolev Institute of Mathematics,
\newline\hphantom{iii} pr. ac. Koptyuga, 4,
\newline\hphantom{iii} 630090, Novosibirsk, Russia}
\email{ivmog@math.nsc.ru}

\address{Faina Ivanovna Solov'eva
\newline\hphantom{iii} Sobolev Institute of Mathematics,
\newline\hphantom{iii} pr. ac. Koptyuga, 4,
\newline\hphantom{iii} Novosibirsk State University,
\newline\hphantom{iii} Pirogova street, 1,
\newline\hphantom{iii} 630090, Novosibirsk, Russia}
\email{sol@math.nsc.ru}%

\thanks{\copyright \ 2020  I. Yu. Mogilnykh, F. I. Solov'eva}
\thanks{\rm This work was supported by the Ministry of Science and Higher Education of Russia (agreement No. 075-02-2020-1479/1).}


\maketitle

\begin{quote}
{\small \noindent{\sc Abstract.} 
 We continue the study of the class of binary extended perfect propelinear codes constructed in the previous paper and consider their permutation automorphism
 (symmetry) groups and Steiner quadruple systems.
We show that the automorphism group of the SQS of any such code
coincides with the permutation automorphism group of the code. In
particular, the SQS of these codes are complete
invariants for the isomorphism classes of these codes. We obtain
a criterion for the point transitivity of the automorphism group
of SQS of proposed codes in terms of $\GL$-equivalence (similar to
EA-type equivalence for permuta\-tions of $F^r$). Based on these results we suggest a
new construction for coordinate transitive and neighbor
transitive extended perfect codes.

\medskip

\noindent{\bf Keywords:}   extended perfect code, concatenation
construction, transitive code, neighbor transitive code,
transitive action,  regular subgroup, isomor\-phism problem,
transitive Steiner quadruple system, coordinate transitive code}
\end{quote}

\section{Introduction}
Propelinear and linear codes share many similar properties and
concepts. Among propelinear codes there are several important
classes of optimal and optimal-related codes such as $Z_4$ and
$Z_2Z_4$-linear, some binary perfect and  extended perfect,
Preparata, Kerdock codes, etc. For more details we refer to the
survey in the introductory part of \cite{MS19}.

The topic of this paper is concerned to the classical problem in
coding theory: can a linear (cyclic, propelinear) code  in a
particular class  be compactly represented, e.g. by its minimum
weight codewords?
 This is true for the Hamming codes, the Reed-Muller codes and the class of extended cyclic codes related to Gold functions \cite{GK}, \cite{MSB}.
 We refer to these works for the surveys on this question.
These type of problems often arise in  testing theory \cite{KL}
and can find an application in cryptography.

Another property of interest is the $k$-transitivity of the
permutation automor\-phism group of a code. An analogous concept
for designs is known as the $k$-point transitivity of their
 automorphism groups. The {\it permutation automorphism group} of a code $C$ denoted by
$\PAut(C)$ (also called the symmetry group of $C$) is the setwise
stabilizer of $C$ in $S_n$, i.e.
$$\PAut(C)=\{\pi: \pi(C)=C\}.$$

 As far as
perfect codes with minimum distance $3$ are concerned, the Ham\-ming
codes are the only known examples with transitive permutation
automorphism group (actually $2$-transitive). 
In the
case of extended perfect codes the only known nonlinear examples are the $Z_4$-linear
extended perfect codes with transitive permuta\-tion
automor\-phism group \cite{Krotov}.
 Other well-known families of binary  codes with the $k$-transitive permuta\-tion automorphism groups include Reed--Muller codes  for $k=3$ (in particu\-lar extended Hamming codes),
extended BCH codes when $k=2$,
classes of extended Preparata codes
 for $k=1$.  In the $q$-ary case there are affine-invariant
codes with $2$-transitive permutation automorphism group \cite{C}.
In the paper we are focused on the case of $1$-transitive action,
and
 a code is called {\it coordinate transitive}  if its
permutation automorphism group acts transitively on the set of its coordinate positions.

 Gillespie and Praeger \cite{GP} suggested the following concept.
A code $C$ is called $k$-{\it neighbor transitive} if its
automorphism group acts transitively on $C_i$, for any $i$, $0\leq
i\leq k$, where $C_i$ is the set of words at distance exactly $i$
from $C$. When $k$ equals the covering radius of the code it is
called completely transitive (originally Sol\'e \cite{Sole}
considered only linear codes). In what follows we call a  Steiner
quadruple system with a $1$-point transitive automorphism group
and a $1$-neighbor transitive code briefly  a point transitive SQS
and a neighbor transitive code respectively.

It is easy to see that a binary transitive code containing the
all-zero vector with $k$-transitive permutation automorphism group
is $k$-neighbor transitive, where $k$ is not larger than the covering radius of the code. Moreover in case of minimum distance
at least three a binary code is neighbor transitive  if and only
if the code is transitive and
 coordinate transitive, see \cite{GP3}.
We see that Reed--Muller and  extended BCH codes \cite{MS} are
$k$-neighbor transitive, where $k$ is $3$ and $2$ respectively and
known extended Preparata codes \cite{Kantor}, \cite{HKCSS} are
neighbor transitive. The Nordstrom--Robinson code, which is
comple\-tely transitive, is an exceptional case, see \cite{GP}.
 Note that for any
known extended Preparata code $P$ which is not the
Nordstrom--Robinson code, its automorphism group also acts
transiti\-ve\-ly on $P_4$, see \cite{MSPPI2017}, but not
transiti\-ve\-ly on $P_2$, see \cite{GP}.

 We note that the
Mollard construction \cite{Mollard} applied to any coordinate
transitive (neighbor transitive) extended perfect codes with
trivial function gives a coordinate transitive (neighbor
transitive) extended perfect code. In the previous paper
\cite{MS19} we constructed a class of extended perfect
propeli\-near codes  utilizing regular subgroups of the general
affine group of the binary vector space in the concatenation
construc\-tion \cite{Sol1981}. We also solved the rank and the
kernel problems for these codes. Any code  of length $n$ and
dimension of the kernel $n-2\log_2n$ from this class was shown to
be a non-Mollard, i.e. can not be obtained by the Mollard
construction \cite{Mollard} with arbitrary function.

 In the current paper we investigate the permutation automorphism groups of the codes from \cite{MS19} and their
 Steiner quadruple systems and find a new  construction for coordinate transitive and neighbor transitive extended perfect
 codes.
 The Steiner quadruple systems of these codes can be described by a particular case of Construc\-tion A*  in the survey of Lindner and Rosa  \cite{LR}.
We prove that the permutation automor\-phism group of the codes
and
 of their Steiner quadruple systems
coincide. Moreover, we prove that the Steiner quadruple systems
are complete invariants for the isomorphism classes of these
codes. As a consequence we are able to find that there are exactly
$64$ isomorphism classes of such codes of length $32$ with the
kernels of dimension $22$. All these codes are non-Mollard.

We obtain an expression for the point transitivity of Steiner
quadruple systems of proposed codes  in terms  similar to extended
affine equivalence for permutations of the vectors of the binary
vector space \cite{CCZ}. For length $n=16$ and $32$ all codes with
the kernels of dimension $n-2\log_2n$ and ranks $n-1$ or $n-2$ 
 have transitive permutation
automorphism 
groups. Exploitation of these  codes having relatively small
kernels gives a series of new neighbor transitive extended perfect
codes of any admissible length $n,n \neq 8, 64$ different from
Mollard and $Z_4$-linear codes. We also note that similar properties hold for a class of propelinear Hadamard codes, which suits 
an informal "duality"\, concept.

\section{Preliminaries}
\subsection{Codes }

We  omit here some definitions that can be found in \cite{MS19}.  By $F^n$ we denote the vector space of dimension $n$ over the Galois field $F$ of two
elements  with respect to the Hamming metric.

Let $\pi$ be a permutation of the coordinate positions of vectors
in $F^n$ that acts as follows:
$\pi(y)=(y_{\pi^{-1}(1)},\ldots,y_{\pi^{-1}(n)})$ for $y\in F^n$.
Let $S_n$ be the symmetric group of the set of coordinate
positions $\{1,\ldots,n\}$. Consider the transformation $(x,\pi)$,
where $x\in F^n$, that maps a binary vector $y$ as
 $$(x,\pi)(y)=x+\pi(y).$$
 The {\it composition} of two
transformations $(x,\pi)$ and $(y,\pi')$ is defined as  follows:
$$(x,\pi)\cdot(y,\pi')=(x+\pi(y),\pi\circ\pi'),$$
where $\pi \circ \pi'$ is the composition of permutations $\pi$
and $\pi'$ acting as
 $$ \pi\circ\pi'(i)=\pi(\pi'(i))$$
 for any $i\in \{1,\ldots,n\}.$  The  {\it automorphism group} $\Aut(F^n)$ of $F^n$ is defined as the
group of all such transformations $(x,\pi)$ with respect to the
composition.  Codes $C$ and $D$ are {\it isomorphic} if there are
$x\in F^n$ and $\pi\in S_n$ such that $x+\pi(C)=D$. We write it as
follows: $C\sim_{(x,\pi)}D$ and $C\sim_{\pi}D$ if $x$ is the
all-zero vector. {\it The automorphism group} $\Aut(C)$ of a code
$C$ is the setwise stabilizer of $C$ in $\Aut(F^n)$. A code $C$ is
{\it transitive} ({\it  propelinear}) if there is a subgroup of
$\Aut(C)$ acting transitively (regularly) on the codewords of $C$.

The Hamming code is an example when the code and its complement
\cite{MSPPI2017} are propelinear codes. For the extended
Nordstrom--Robinson code we have the following result obtained by
a computer (the results for $N$ and $N_4$ are known
\cite{MSPPI2017}).

\begin{proposition}
Let $N$ be the extended Nordstrom--Robinson code, $N=N_0$, $N_1$,
$N_2,$ $N_3$ and $N_4$ be the distance partition with respect to
$N$. For any $i \in\{0,1,3,4\}$ there is a subgroup of $\Aut(N)$
that acts regularly on $N_i$ and there is no such subgroup for
$N_2$.
\end{proposition}

Let us consider one simple property of transitive codes which will
be useful for our further investigations.

\begin{proposition}\label{Prop1}
Let $C$ and $D$ be  transitive codes of length $n$ containing the
all-zero vector such that $C\sim_{(x,\pi)} D$. Then $C\sim_{\pi'}
D$ for some permutation $\pi'$ of $S_n$.
\end{proposition}

\begin{proof}
Let $C\sim_{(x,\pi)} D$, i.e. we have $\pi(C)=x+D$. Since $C$ and
$D$ contain all-zero vectors, the latter equality implies that $x$
is in $D$. Taking into account that $D$ is transitive we have
$x+D=\pi^{\prime\prime}(D)$ for some permutation
$\pi^{\prime\prime}\in S_n$. Hence we have
$C=\pi^{-1}\pi^{\prime\prime}(D)$.
\end{proof}

Let the coordinates of the vector space $F^{2^r}$ be indexed by
the vectors of $F^r$. The all-zero vector is denoted by ${\bf 0}$
and its length will always be clear from the context.
 Define an extended Hamming code of length $2^r$ as follows:
 $${\mathcal H}=\{x\in F^{2^r}:\sum_{a: x_a=1} a={\bf 0},\, \wt(x)\equiv 0(\mbox{mod }2)\}.$$

The concatenation of two vectors $x\in F^{r'}$ and $y\in F^{r''}$
is denoted by $x|y$. For codes $C$ and $D$  by $C\times D$ we
denote the code $\{x|y:x\in C,\, y \in D\}$. Let $\pi'$, $\pi''$
be permutations on the vectors of $F^{r'}$ and $F^{r''}$
respectively. By $\pi'|\pi''$ we denote the permutation  on the
vectors of $F^{r'+r''}$ acting on the concatenations $x|y$ of the
vectors $x\in F^{r'}$ and $y\in F^{r^{\prime\prime}}$ as follows:
$(\pi'|\pi'')\,(x|y)=\pi'(x)|\pi''(y)$. In particular, if 
$\pi'$ and
$\pi''$ are permutations of the coordinate positions of $F^{r'}$
and $F^{r''},$
 then $(\pi'|\pi'')$ is a permutation of
the coordinate positions of   $F^{r'+r''}$.


 Let $e_{ a}$ be the vector in $F^{2^r}$ with the
only one nonzero position indexed by a vector $a\in F^r$.

Consider the following  particular case of the concatenation
construction  for extended perfect codes \cite{Sol1981}:
\begin{equation}\label{Soloveva_codes}
S_{\tau}=\bigcup_{ a\in F^r} ({\mathcal H}+e_{ a}+e_{\bf 0})\times
({\mathcal H}+e_{\tau(a)}+e_{\tau({\bf 0})}),\end{equation} where
$\tau$ is a permutation of the vectors of $F^{r}$. In throughout
of what follows we suppose that $\tau$ fixes ${\bf 0}$, so the
code $S_{\tau}$ contains the all-zero vector. Note that in
\cite{MS19} the code 
$S_{\tau}$ is denoted by $S_{{\mathcal H},\tau}$.

Denote the {\it general linear group} that consists of the
nonsingular $r\times r$ matrices  over  $F$  by  $\GL(r,2)$.
 Consider an
affine transformation $( a,M)$, $a\in F^r, M \in
\GL(r,2)$. Its action on $F^r$ is defined as
\begin{equation}\label{actionGA}(a, M)(b) = a + Mb,\end{equation} $b \in F^r$.  The {\it general affine
group} of the space $F^r$ whose elements are $\{(a,M):a \in F^r,
M\in \GL(r,2)\}$ with respect to the composition is denoted by
$\GA(r,2)$. The group $\GA(r,2)$ naturally acts on the positions
indexed by the vectors of $F^r$. Let $\sigma_{a,M}$ denote the
permutation on the positions of $F^{2^r}$ that corresponds to the
affine transformation $(a,M)$. It is well-known that the
automorphism group of an extended Hamming code is isomorphic to
the general affine group, i.e.
$$\Aut({\mathcal H})=\{\sigma_{a,M}:(a,M)\in \GA(r,2)\}.$$
For $a \in F^r$, $M \in \GL(r,2)$ we denote
the linear map $\sigma_{{\bf 0},M}$ by $\sigma_M$ and denote $C
\sim_M D$ if $C\sim _{\sigma_M}D$. When $M$ is the identity matrix
we denote the translation $\sigma_{a,M}$ by $\sigma_{a}$. We have
the following result.
\begin{proposition}\label{prop_affine} Let $\tau$ be a permutation on the vectors of $F^r$, $\tau({\bf 0})={\bf 0}$.
 The code $S_{\tau}$ is an extended Hamming code if and only if $\tau$ is $\sigma_M$ for some $M\in \GL(r,2)$.
\end{proposition}

\subsection{Steiner quadruple systems}

A {\it Steiner quadruple system of order $n$} (briefly SQS if we
know the order by the context) is a set of quadruples (subsets of
size $4$) of a point set of size $n$ where every three points from
the point set are contained in exactly one quadruple.
It is well-known that the supports of
the codewords of weight $4$ of any extended perfect code containing
the all-zero vector form a Steiner quadruple system.

The {\it automorphism group} of a Steiner quadruple system $Q$ of
order $n$, denoted by $\Aut(Q)$ is the setwise stabilizer of the
quadruples of $Q$ in the symmetric group of the pointset of $Q$.
Steiner quadruple systems $Q$ and $Q'$ are {\it isomorphic} if
there is a bijection $\pi$ between their point sets that sends the
quadruples of $Q$ to these of $Q'$. In this case we write
$Q\sim_{\pi} Q'$. A Steiner quadruple system  is called {\it point
transitive} if its permutation automorphism group acts
transitively on the set of its points.

By an {\it affine Steiner quadruple system}, briefly {\it affine
SQS}, we mean the Steiner quadruple system of an extended Hamming
code.  It is well-known that the Hamming codes of fixed length are
unique up to a permutation  and therefore any extended Hamming
code is spanned by its codewords of weight 4 by Glagolev theorem. We conclude that all
affine Steiner quadruple systems are isomorphic.

We now describe the Steiner quadruple system of the code
$S_{\tau}$ (see (\ref{Soloveva_codes})), which we denote by
$SQS_{\tau}$. For any \ $x,y \in F^{2^r}$ the support of \, $x|y
\in F^{2^{r+1}}$ \, is
denoted by the ordered pair $(\supp(x),\supp(y))$.  
  The pointset of $SQS_{\tau}$ consists of  $(\{a\},\emptyset),\, (\emptyset,\{a\})$, for all $a$ in $F^r$ and
  $SQS_{\tau}$ is $Q_0\cup Q_1 \cup Q_{\tau}$, where
$$Q_0=\{(\{a,b,c,d\},\emptyset): a,b,c,d \mbox{ are  pairwise distinct vectors of } F^r,  a+b+c+d={\bf 0}\},$$
$$Q_1=\{(\emptyset, \{a,b,c,d\}): a,b,c,d \mbox{ are  pairwise distinct vectors of } F^r, a+b+c+d={\bf 0}\},$$
$$Q_\tau=\{(\{a,c\},\{b,d\}):a,b,c,d \in F^r, \tau(a+c)=b+d \neq {\bf 0}  \}.$$
For a pair of distinct vectors $a, b\in F^r$ we define the
following sets of quadruples:
$$Q_0(a,b)=\{ (\{a,b,c,d\},\emptyset): c, d \in F^r,  c\neq d, a+b+c+d={\bf 0}\}$$
$$\cup  \{ (\{a,b\},\{c,d\}): c,d \in F^r, \tau(a+b)=c+d\neq {\bf 0}\},$$
$$Q_1(a,b)=\{ (\emptyset,\{a,b,c,d\}): c,d  \in F^r,  c\neq d, a+b+c+d={\bf 0}\},$$
$$\cup  \{ (\{c,d\},\{a,b\}): c, d \in F^r, \tau(c+d)=a+b \neq {\bf 0}\},$$
and for a pair of vectors $a, b\in F^r$ we define
$$Q_\tau(a,b)=\{ (\{a,c\},\{b,d\}): c, d \in F^r, \tau(a+c)=b+d\neq {\bf 0}\}.$$

In other words, $Q_0(a,b)$, $Q_1(a,b)$ and $Q_{\tau}(a,b)$ are the
sets of all quadruples in $SQS_\tau$, that contain the pairs of
points $(\{a\},\emptyset)$ and $(\{b\},\emptyset)$;
$(\emptyset,\{a\})$ and $(\emptyset,\{b\})$;  $(\{a\},\emptyset)$
and $(\emptyset,\{b\})$ respectively.


\subsection{Coordinate transitive Mollard extended perfect codes}

 Let $C$ and $D$ be exten\-ded perfect codes of lengths $t$ and $m$ respectively, $\phi$ be a function from $C$ to  $F^m$.
The coordinates of the Mollard extended perfect code are pairs
$(r,s)$, where $r\in \{1,\ldots,t\}$ and $s\in \{1,\ldots,m\}$.
For a vector
$z=(z_{11},\ldots,z_{1m},\ldots,z_{t1},\ldots,z_{tm})$ in $F^{tm}$
we consider
$$p_1(z)=(\sum_{s=1}
^{m}z_{1s},\ldots,\sum_{s=1}^{m}z_{ts}), \,\,\,\,\,
p_2(z)=(\sum_{r=1}^{ t}z_{r1},\ldots,\sum_{r=1}^{t}z_{rm}).$$
 The Mollard extended perfect code is defined as follows: $$M(C,D)=\{z\in  F^{tm}: p_1(z)\in C, \,\,\,\,  p_2(z)\in
 \phi(p_1(z))+D\}.$$ It is not hard to see that $M(C,D)$ is the extension of the  Mollard perfect code by overall parity check (the original construction \cite{Mollard} used perfect codes).
Analogously to the previous works \cite{Stransitive}, \cite{MSM},
\cite{BMRS} we now consider embeddings of the permutation
automorphism groups of $C$ and $D$ into that of the Mollard code
with the all-zero function $\phi$. For  permutations $\pi\in
\PAut(C)$ and $\pi'\in \PAut(D)$ we define $Dub_1(\pi)$ and
$Dub_2(\pi')$ acting on the positions of $M(C,D)$ as follows:
$$Dub_1(\pi)(r,s)=(\pi(r),s), \,\,\,\, Dub_2(\pi')(r,s)=(r,\pi'(s)).$$
It is not hard to see that $Dub_1(\pi)$ and $Dub_2(\pi')$ are
permutation automorphisms of $M(C,D)$, see \cite{Stransitive},
\cite{MSM}, \cite{BMRS}.
\begin{proposition}
If $C$ and $D$ are coordinate transitive (neighbor transitive)
exten\-ded perfect codes then $M(C,D)$ with the all-zero function
$\phi$ is a  coordinate transitive (neighbor transitive) extended
perfect code.
\end{proposition}
\begin{proof}
If $\PAut(C)$ and $\PAut(D)$ act transitively on the coordinates of $C$ and $D$, then the group generated 
 by $\{Dub_1(\pi):\pi \in \PAut(C)\}$ and $\{Dub_2(\pi'):\pi' \in \PAut(D)\}$ acts transitively on the coordinates of $M(C,D)$.
 Moreover, if $C$ and $D$ are transitive then $M(C,D)$ is transitive, see \cite{Stransitive}.
 We conclude that $\Aut(M(C,D))$ acts transitively on the codewords of $M(C,D)$ and its coordinates, i.e. $M(C,D)$ is neighbor transitive code.
\end{proof}

We say that a code is {\it non-Mollard} if it is not isomorphic to
a Mollard code with arbitrary $\phi$. In Section 4 we obtain an
infinite series of  neighbor transitive extended perfect
non-Mollard codes.

\section{The automorphism groups of a class of Steiner quadruple systems}




The following lemma reveals the discrepancy in the "linearity" \,
of the quadruples from $Q_0(a,b)$ and $Q_1(a,b)$ and the nonlinearity of
that of $Q_\tau(a,b)$.

\begin{lemma}\label{L1}
Let $\tau$ be a  permutation on the vectors of $F^r$ that fixes
${\bf 0}$ and $SQS_\tau$ be   non-affine. Then

1. For any  $a, b\in F^r$, $a\neq b$ the symmetric difference of
any pair of distinct quadruples from $Q_0(a,b)$  $(Q_1(a,b))$  is
in $SQS_\tau$.

2. For any  $a, b\in F^r$ there are distinct quadruples in
$Q_\tau(a,b)$ whose symmetric difference is not in $SQS_{\tau}$.

\end{lemma}
\begin{proof}

1. Consider two arbitrary quadruples from $Q_0(a,b)$. We have
three different cases.

\medskip
Case A. Let $(\{a,b,c,d\},\emptyset)$ and
$(\{a,b,e,f\},\emptyset)$ be distinct quadruples of $Q_0$. Then
their symmetric difference is obviously in $Q_0$, because these
are the supports of the codewords of weight four of the extended
Hamming code ${\mathcal H}$ of length $2^r$.

\medskip
Case B. Let $(\{a,b\},\{c,d\})$ and $(\{a,b\},\{e,f\})$ be
distinct quadruples from $Q_0(a,b)$, where $\tau( a+ b)=c+d= e+
f$. Then their symmetric difference $(\emptyset,\{c,d,e,f\})$ is
in $Q_1$, because c+d+e+f={\bf 0}.

\medskip
Case C. Let $(\{a,b,c,d\},\emptyset)$ and $(\{a,b\},\{e,f\})$ be
from $Q_0(a,b)$, where $a+b=c+d$ and $\tau(a+b)=e+f$. Then their
symmetric difference $(\{c,d\},\{e,f\})$ is in $Q_\tau$, because
$\tau(c+d)=\tau(a+b)=e+f$.

Similar considerations hold for any pair of quadruples of
$Q_1(a,b)$.

\medskip
2. Consider any two distinct quadruples from $Q_\tau(a,b)$:
\begin{equation}\label{Quad1}
(\{a,c\},\{b,d\}), (\{a,e\},\{b,f\}).
\end{equation}
By the definition of $Q_{\tau}(a,b)$ we have:
\begin{equation}\label{Quad2}
\tau(a+ c)=b+d, \, \tau(a+e)=b+f.
\end{equation}

The symmetric difference of 
the quadruples from (\ref{Quad1}) is
$(\{c,e\},\{d,f\}).$ It belongs to $Q_\tau$ if and only if
$\tau(c+e)=d+f$, which  taking into account (\ref{Quad2}) holds if
and only if $\tau(c+e)=\tau(a+c)+ \tau(a+e)$. It is easy to see
that the latter equality holds for any distinct $c$ and $e$
different from $ a$ if and only if $\tau(c+e)=\tau(c)+\tau(e)$. We
conclude that the symmetric difference of any pair of distinct
quadruples from $Q_\tau(a,b)$ is in $SQS_\tau$ if and only if
$\tau$ is a linear mapping, i.e. $\tau=\sigma_M$ for some $M\in
\GL(r,2)$, so $SQS_{\tau}$ is affine by Proposition
\ref{prop_affine}.
\end{proof}

\begin{proposition}\label{prop1}
Let $\tau$ be a permutation on the vectors of
 $F^r$, $\tau({\bf 0})={\bf 0}$. Then for any $a,b \in F^r$ we have $(\sigma_{a}|\sigma_{b})\in \Aut(SQS_{\tau})$. In particular,
$\Aut(SQS_\tau)$ either acts transitively on the points of
$SQS_{\tau}$ or has two orbits of points that are $(\{a:a\in
F^r\},\emptyset)$ and $(\emptyset,\{a:a\in F^r\})$.

\end{proposition}

\begin{proof}
Obviously $(\sigma_a|\sigma_b)$ fixes the quadruples from $Q_0$
and $Q_1$. Let $(\{u,v\},\{g,f\})$ be a quadruple of $Q_{\tau}$,
so $\tau(u+v)=g+f$. From the latter equality, we have
that $(\sigma_a|\sigma_b)(\{u,v\},\{g,f\})=(\{u+a,v+a\},\{g+b,f+b\})$
 is in $Q_\tau$.
\end{proof}

 We
denote by $\xi$ the permutation that swaps the points (coordinate positions)
$(\{a\},\emptyset)$ and $(\emptyset, \{a\})$ for any $a\in F^r$. 


\medskip
\noindent
 {\bf Remark 1.} Let $A$ be a matrix from $GL(r,2)$ and let
$\tau$ be a permutation of the vectors of $F^r$. In below we use
$A\tau$ and $\tau A$ to denote the permutations of the vectors of
$F^r$ that are compositions of the linear mapping $A$ and $\tau$,
i.e $\sigma_A\tau$ and $\tau\sigma_A$ respectively.

\begin{theorem}\label{Theo1}
Let $\tau$ and $\tau'$ be permutations on the vectors of
 $F^r$ that fix ${\bf 0}$. Two $SQS_{\tau}$ and $SQS_{\tau'}$ of order $2^{r+1}$ satisfy  $SQS_{\tau}\sim_\pi SQS_{\tau'}$ if and only if
 one of the following conditions holds:

1) both $SQS_{\tau}$ and $SQS_{\tau'}$ are affine;

2)  the permutation $\pi$ is equal to $(\sigma_{a,A}|\sigma_{
b,B})$, where $( a,A)$, $( b,B)\in \GA(r,2)$ and
$\tau'=B\tau{A}^{-1}$;

3) the permutation $\pi$ is equal to
$(\sigma_{a,A}|\sigma_{b,B})\xi$, where $( a,A),$ $( b,B)\in
\GA(r,2)$ and $\tau'=B\tau^{-1}{A}^{-1}$.\end{theorem}

\begin{proof}

Sufficiency. Case 1.  It is well known that if both $SQS_{\tau}$ and $SQS_{\tau'}$ are affine, then they are isomorphic.

We can factor out the permutations $(\sigma_a|\sigma_b)$  as they are common automorphisms for $SQS_\tau$ and $SQS_{\tau'}$ by Proposition \ref{prop1}.
 Throughout cases 2 and 3 below we consider
 an arbitrary
quadruple $(\{e,f\},\{c,d\})$ from $Q_{\tau}$ where $e,f\in F^r$
are distinct  and  fulfill
\begin{equation}\label{eqq}\tau(
e+ f)=c+d.\end{equation}

Case 2.
We show that \begin{equation}\label{eq00}(\sigma_A|\sigma_B)(SQS_{\tau})=SQS_{B\tau{A}^{-1}},\end{equation} for any $A,B\in \GL(r,2)$.

Obviously, $(\sigma_A|\sigma_B)$ fixes $Q_0$ and $Q_1$.
 We have
$$(\sigma_A|\sigma_B)(\{e,f\},\{c,d\})=(\{Ae,Af\},\{Bc,Bd\}).$$ From (\ref{eqq}) this vector is in $Q_{B\tau A^{-1}}$ because $B\tau A^{-1}(Ae+Af)=B\tau(e+f)=Bc+Bd$, so (\ref{eq00}) holds.

Case 3. We are to show that \begin{equation}\label{eq01}(\sigma_A|\sigma_B)\xi(SQS_{\tau})=SQS_{B\tau^{-1}{A}^{-1}},\end{equation} for any $A,B\in \GL(r,2)$.

We note that $\xi(SQS_{\tau})=SQS_{\tau^{-1}}$. Indeed, it is
obvious that $\xi(Q_0)=Q_1$. Let $(\{e,f\},\{c,d\})$ be from
$Q_{\tau}$. Then $\xi(\{e,f\},\{c,d\})=(\{c,d\},\{e,f\})$ is in
$SQS_{\tau^{-1}}$ because $\tau^{-1}(c+d)$ is $e+f$ by
(\ref{eqq}).

From $\xi(SQS_{\tau})=SQS_{\tau^{-1}}$ and (\ref{eq00}) we have
that
$$(\sigma_A|\sigma_B)\xi(SQS_{\tau})=(\sigma_A|\sigma_B)(SQS_{\tau^{-1}})=SQS_{B\tau^{-1}{A}^{-1}},$$
i.e.   (\ref{eq01}) holds.

Necessity. It is enough to consider non-affine $SQS_{\tau}$ and
$SQS_{\tau'}$, $SQS_{\tau}\sim_\pi SQS_{\tau'}$. By Lemma \ref{L1}
for any $a,b\in F^r$ the subsets $Q_0(a,b)$ and $Q_1(a,b)$ of
$SQS_{\tau'}$ could not be $\pi(Q_{\tau}(a,b))$. In other words,
$\pi(\{a:a\in F^r\},\emptyset)$ is either $(\{a:a\in
F^r\},\emptyset)$ or $(\emptyset,\{a:a\in F^r\})$, so $\pi(Q_0\cup
Q_1)$=$Q_0\cup Q_1$. It is easy to see that $\pi(Q_0\cup
Q_1)$=$Q_0\cup Q_1$ if and only if   $\pi$ is
$(\sigma_{a,A}|\sigma_{b,B})\xi^t$, where $(a,A), (b,B)\in
\GA(r,2)$, $t\in \{0,1\}$.

\end{proof}

\begin{corollary}\label{CoroSeq}
Let $\tau$ be  the permutation on the vectors of $F^r$
that fixes ${\bf 0}$. 
 The Steiner quadruple system $SQS_{\tau}$ is point transitive if
and only if $\tau^{-1}\in \GL(r,2)\tau \GL(r,2)$.
\end{corollary}
\begin{proof}
By Proposition \ref{prop1} for any permutation $\tau$ the points
of $(\{a:a\in F^r\},\emptyset)$ and $(\emptyset, \{a:a\in F^r\})$
are in one orbit of $\Aut(SQS_{\tau})$. So, $SQS_{\tau}$ is point
transitive if and only if there is $\pi\in \Aut(SQS_{\tau})$, such
that $\pi(\emptyset, \{a\})=\pi(\{b\},\emptyset )$ for some $a,
b\in F^r$. We have the following description for all such $\pi$
and $\tau$ from Theorem 1:
$$\{(\sigma_{A}|\sigma_{B})\xi: A,B\in GL(r,2) \} \,\, \mbox {and } \, \tau^{-1}\in \GL(r,2)\tau \GL(r,2).$$\end{proof}
\begin{theorem}\label{Theo2}
Let $\tau$ and $\tau'$ be  permutations on the vectors of $F^r$
that fix ${\bf 0}$. For a permutation $\pi$ we have
$SQS_{\tau}\sim_\pi SQS_{\tau'}$  if and only if $S_{\tau}\sim_\pi
S_{\tau'}$.

\end{theorem}

\begin{proof}
The sufficiency is clear, we show the necessity. Each code $S_{\tau}$ could be
represented as follows:
$$S_{\tau}=\{x+y: supp(x)\in SQS_{\tau}, y\in {\mathcal H}\times {\mathcal H} \}.$$
From the proof of Theorem \ref{Theo1} we see that the permutation
$\pi$ such that $SQS_{\tau}\sim_\pi SQS_{\tau'}$ preserves
$Q_0\cup Q_1$. We conclude that $\pi$ also preserves ${\mathcal
H}\times {\mathcal H}$, which is spanned by the characteristic
vectors of the quadruples from $Q_0 \cup Q_1$. Therefore we have
$S_{\tau}\sim_\pi S_{\tau'}$.\end{proof}

\begin{corollary}\label{C3}
The groups  $\PAut(S_{\tau})$ and $\Aut(SQS_{\tau})$ are
isomorphic.
\end{corollary}
The necessity of the following statement is by Proposition \ref{Prop1} and the
sufficiency
 is by Theorem \ref{Theo2}.
\begin{corollary}\label{C2}
Let $\tau$ and $\tau'$ be  permutations on the vectors of
 $F^r$ that fix ${\bf 0}$. Let $S_{\tau}$ and $S_{\tau'}$ be transitive codes of length $2^{r+1}$. Then
$S_{\tau}\sim_{(x,\pi)} S_{\tau'}$ for a vector $x$ and a
permutation $\pi$ if and only if $SQS_{\tau}\sim_{\pi'}
SQS_{\tau'}$ for some permutation $\pi'$.
\end{corollary}


\section{An infinite series of coordinate transitive and  neighbor transitive extended perfect codes}
We recall some concepts from \cite{MS19}. A subgroup $G$ of the
general affine group $\GA(r,2)$ is called {\it regular} if it is
regular with respect to the action (\ref{actionGA}) on the vectors
of $F^r$. By the definition for any regular subgroup $G$ of the
group $\GA(r,2)$ and any $a\in F^r$ there is a unique affine
transformation that maps ${\bf 0}$ to $a$, which  we denote by
$g_a$. Obviously, $g_a$ is $(a,M)$ for some matrix  $M$ in
$\GL(r,2)$. Let  $T$ be an automorphism of a regular subgroup  $G$
of the group $\GA(r,2)$. By $\tau$ we denote the permutation on
the vectors of $F^{r}$ {\it induced} by the action of the
automorphism  $T$, i.e.
$$T(g_a)=g_{\tau(a)}.$$
Obviously we always have $\tau({\bf 0})={\bf 0}$.
 The following class of propelinear codes was obtained in \cite{MS19}.

\begin{theorem}\label{Theoprev}\cite{MS19} Let $G$  be a regular subgroup of $\GA(r,2)$ and $\tau$ be the permutation induced by an automorphism of $G$.
The following hold:

1. The code $S_{\tau}$ is a propelinear extended perfect binary
code of length $2^{r+1}$.

2. Let $\tau'$ be the permutation induced by an automorphism of
$\GA(r',2)$. Then $(\tau|\tau')$ is the  permutation induced by an
automorphism of $\GA(r+r',2)$, in particular $S_{\tau|\tau'}$ is a
propelinear extended perfect binary code of length $2^{r+r'+1}$.

3. If $S_{\tau}$ has the kernel of dimension $2^{r+1}-2r-2$, then
$S_{\tau}$ is not a Mollard code. If additionally $\tau'$ is a
permutation of the vectors of $F^{r'}$, $\tau'({\bf 0})={\bf 0}$
and $S_{\tau'}$ has the kernel of dimension $2^{r'+1}-2r'-2$ then
$S_{\tau|\tau'}$ is a non-Mollard code with the kernel of
dimension $2^{r'+r+1}-2(r'+r)-2$.

4. Propelinear non-Mollard codes $S_{\tau}$ of length
$2^{r+1}$ with the kernel of dimension $2^{r+1}-2r-2$  exist for
any $r, r\geq 3$.

\end{theorem}

We first consider neighbor transitive extended perfect codes of
small length. Then we iteratively  construct coordinate transitive
(neighbor transitive) extended perfect codes of any admissible
length not equal to 64.

\begin{theorem}\label{Theo16}
For any permutation $\tau$ on the vectors of $F^3$,   $\tau({\bf
0})={\bf 0}$, the extended perfect code $S_{\tau}$ of length $16$ is
  neighbor transitive. There are exactly four isomor\-phism
classes for these codes and they are characterized by their ranks
that take values $11$, $12$, $13$, $14$. The code of rank $14$ has dimension of kernel $8$.
\end{theorem}
\begin{proof}
We now show the following natural description of the double cosets
of $\GL(r,2)$ in terms of intersections of Hamming codes.
\begin{lemma}\label{Propeq}
Let $\tau$ and $\tau'$ be permutations on the vectors of $F^r$
that fix ${\bf 0}$.
We have $\tau'\in \GL(r,2)\tau\GL(r,2)$  if and only if  $\tau({\mathcal H})\cap {\mathcal H}\sim_{A}\tau'({\mathcal H})\cap {\mathcal H}$ for some $A \in \GL(r,2)$.
\end{lemma}
\begin{proof}
Let $\tau'$ be $A^{-1}\tau B$, $A$ and $B$ be in $\GL(r,2)$. We have
$A({\mathcal H})=B({\mathcal H})={\mathcal H}$ and the
following:
$$\tau'({\mathcal H})\cap {\mathcal H}=A^{-1}\tau B({\mathcal H})=A^{-1}\tau({\mathcal H})\cap {\mathcal H}\sim_{A}\tau({\mathcal H})\cap {\mathcal H}.$$\end{proof}
For the extended Hamming code ${\mathcal H}$ of length 8 and
arbitrary permutation $\tau$ on the vectors of $F^3$ there are
exactly four possible values for the dimension of $\tau({\mathcal
H})\cap {\mathcal H}$: $1,\,2,\,3,\,4$. Moreover, it is easy to
see that for any fixed dimension $k\in \{1,2,3,4\}$ there is
$\tau$ such that the code $\tau({\mathcal H})\cap {\mathcal H}$
has dimension $k$ and is unique up to a permutation of $\GL(3,2)$.

In view of Lemma \ref{Propeq} we see that there are exactly four
double cosets by $\GL(3,2)$  in the group of all permutations of
$F^3$ that fix ${\bf 0}$. By Theorems \ref{Theo1} and \ref{Theo2}
we conclude that there are not more than 4 isomorphism classes for
the codes. From \cite{MS19}, see e.g. Table 1, there are exactly
$4$ different values $\{11,12,13,14\}$ for the ranks for the
propelinear codes $S_{\tau}$ of length 16 where $\tau$ runs
through the permutations induced by the automorphisms of the
regular subgroups of $\GA(3,2)$. Therefore we have exactly $4$
isomorphism classes and each code is propelinear.

 Obviously for any permutation $\tau$ on the vectors of $F^3$,   $\tau({\bf
0})={\bf 0}$ we have that  $\tau^{-1}({\mathcal H})\cap {\mathcal
H}$ has the same dimension as $\tau({\mathcal H})\cap {\mathcal
H}$. Taking into account what was declared at the beginning of the
proof, there is $A\in\GL(3,2)$ such that $\tau^{-1}({\mathcal
H})\cap {\mathcal H}\sim_A\tau({\mathcal H})\cap {\mathcal H}$.
Then using Lemma \ref{Propeq} we see that $\tau^{-1}$ is in
$\GL(3,2)\tau\GL(3,2)$. By Corollaries \ref{CoroSeq} and \ref{C3}
we conclude that the code $S_{\tau}$ is coordinate transitive for
any $\tau$. Since $S_{\tau}$ is transitive, it is  neighbor
transitive.
The code of rank $14$ has the dimension of the kernel $8$ (see \cite{MS19}) and therefore is a non-Mollard code by Theorem \ref{Theoprev}.
\end{proof}

\begin{theorem}\label{Theo32}
There are exactly $64$ isomorphism classes of propelinear extended
perfect codes $S_{\tau}$ of length $32$ with the dimension of kernel
$22$, where $\tau$ is a permuta\-tion induced by an automor\-phism
of a regular subgroup of $GA(4,2)$. All these codes are  neighbor
transitive non-Mollard.
\end{theorem}
\begin{proof}
 Let $\tau$  run through the permutations induced
by the automorphisms of the regular subgroups 
 of $\GA(4,2)$ such
that the corresponding codes $S_{\tau}$ of length $32$ have kernels
of dimension $22$. These codes are transitive non-Mollard codes by
Theorem \ref{Theoprev}. Their isomorphism problem is equivalent to
the isomorphism problem of their SQS's by Corollary \ref{C2}. This
problem for SQS of order $32$ with $1240$ quadruples could be solved
by a computer. By MAGMA we find that there are exactly $64$ isomorphism
classes of the SQS's of these codes.
 Moreover, again by a computer we
see that for any such code $S_{\tau}$ the permutations
 $\tau$ and $\tau^{-1}$ are in the same double coset of $\GL(4,2)$. By Corollary \ref{CoroSeq} we see that all
these codes are neighbor transitive.
\end{proof}

The following could be considered as a direct product construction
for permuta\-tions providing coordinate transitive and neighbor
transitive codes.

\begin{theorem}\label{Theo6}
Let $\tau$ and $\tau'$ be  permutations on the vectors of
 $F^r$ and $F^{r'}$ that fix all-zero vectors. We have the
 following:

1. If SQS$_\tau$ and SQS$_{\tau^{\prime}}$ are point transitive
Steiner quadruple systems of orders $2^{r+1}$ and $2^{r^{\prime}+1}$
respectively then SQS$_{\tau|\tau^{\prime}}$ is a point
transitive Steiner quad\-ruple system of order $2^{r + r^{\prime}+1}.$

2. If $S_{\tau}$  and $S_{\tau'}$ are  coordinate transitive codes
then the code $S_{\tau|\tau'}$ is coordinate transitive.

3. If $\tau$ and $\tau'$ are induced by automorphisms of regular
subgroups of $\GA(r,2)$ and $\GA(r',2)$ respectively and $S_{\tau}$,
$S_{\tau'}$ are coordinate transitive codes then the code
$S_{\tau|\tau'}$ is neighbor transitive.
\end{theorem}
\begin{proof}
1. By Corollary \ref{CoroSeq} we have $\tau^{-1}=A\tau B$ and
$\tau^{\prime^{-1}}=A^{\prime}\tau^{\prime} B^{\prime}$ for
appropriate matrices $A,\,B$ in $\GL(r,2)$ and $A^{\prime}, \,
B^{\prime}$ in $\GL(r^{\prime},2)$. Therefore
$$(\tau|\tau^{\prime})^{-1}=\left(%
\begin{array}{cc}
  A &  {\bf 0}_{r,r^{\prime}} \\
 {\bf 0}_{r^{\prime}, r} & A^{\prime}\\
\end{array}%
\right)\circ (\tau|\tau^{\prime})\circ\left(%
\begin{array}{cc}
  B &  {\bf 0}_{r,r^{\prime}} \\
 {\bf 0}_{r^{\prime}, r} & B^{\prime}\\
\end{array}%
\right),$$ here ${\bf 0}_{r,r^{\prime}}$ and
 ${\bf 0}_{r^{\prime}, r}$ are the all-zero $r\times r^{\prime}$ and  $r^{\prime}\times
 r$
matrices respectively and $\circ$ denotes the composition of the
permutations of the vectors from $F^{r+r'}$. Hence according to 
Corollary \ref{CoroSeq} we see that  SQS$_{\tau|\tau^{\prime}}$ is
point transitive.

2. Follows from Corollary \ref{C3} and the first statement of this
theorem.

3. From the second statement of the current theorem we see that
the code $S_{\tau|\tau'}$ is  coordinate transitive. From the
second statement of Theorem \ref{Theoprev} we see that the code
$S_{\tau|\tau'}$ is transitive, so it is neighbor transitive.
\end{proof}

\begin{theorem}\label{Theo7}
For any $r\geq 3$, $r\neq 5$ there exists a  neighbor transitive
extended perfect code $S_{\tau|\tau^{\prime}}$ of length $2^{r+1}$
that is non-Mollard.
\end{theorem}
\begin{proof}
 The proof will be  done by induction. For the induction
base we use  Theorems \ref{Theo16} and  \ref{Theo32} 
 with the
exception of the codes of length 64. Let $\tau$ and
$\tau^{\prime}$ be permutations induced by automorphisms of
regular subgroups of $\GA(r_1,2)$ and $\GA(r_2,2)$, $r=r_1 +r_2$
such that neighbor transitive codes $S_{\tau}$ and
$S_{\tau^{\prime}}$ are of length $2^{r_1+1}$ and $2^{r_2+1}$,
with kernels of dimension $2^{r_1+1}-2r_1$ and $2^{r_2+1}-2r_2$
respectively.
 Applying Theorem \ref{Theo6}
 we obtain a neighbor transitive
extended perfect code $S_{\tau|\tau^{\prime}}$ of length
$2^{r+1}$
that is non-Mollard by the third statement of Theorem
\ref{Theoprev}.
\end{proof}

\noindent {\bf Remark 2.} Analogously to the approach described in
 Theorem \ref{Theo7} new
class of codes could be separated from $Z_4$-linear codes via
ranks. From Theorem \ref{Theo16} we see that the non-Mollard code
has the dimension of the kernel $8$ and rank $14$, whereas at
least one of the codes classified in Theorem \ref{Theo32} has rank
$30$ (see \cite{MS19}). By induction as in Theorem \ref{Theo7} we
obtain  non-Mollard codes with large (prefull or preprefull)
ranks, see \cite[Corollary 2]{MS19}.
 We conclude that the obtained neighbor transitive codes are inequivalent to $Z_4$-linear codes from \cite{K1} and the Mollard codes.

\medskip

\noindent {\bf Remark 3.} 
We note that almost all results of
this work and \cite{MS19} hold for the following construction of
Hadamard codes.

Let us consider the following representation of a linear
$(2^r,r+1,2^{r-1})$ Hadamard code. We index the positions of
$F^{2^r}$ by the binary vectors of length $r$. Let $C_a$,
$a\in
F^r$, be the code
of length $2^r$ with the codewords that have the
following supports:
 $$\{x\in F^r: \, <x,a>=0 \}, \,  \{x\in F^r: \, <x,a>=1\},$$ where
$<\cdot,\cdot>$ is a scalar product.
Obviously the set $\bigcup_{a\in F^r} C_a$ is a Hadamard code
of length $2^r$.

Let $\tau$ be a permutation induced by an automorphism of a regular subgroup of $\GA(r,2)$. Then the following code
$$A_{\tau}=\bigcup_{a\in F^r} C_a\times C_{\tau(a)} $$
is a propelinear Hadamard code of length $2^{r+1}$. Analogously to the proof of Theorem \ref{Theo1} we have that the codes $A_{\tau}$
and $A_{\tau'}$ are isomorphic if and only if $\tau$ or $\tau^{-1}$ is in  $\GA(r,2)\tau'\GA(r,2)$. Thus the isomorphism classes of these Hadamard codes and those of the extended
perfect codes are in one-to-one correspondence.

\end{document}